\documentclass{article}

\usepackage[english]{babel}

\usepackage[letterpaper,top=2cm,bottom=2cm,left=3cm,right=3cm,marginparwidth=1.75cm]{geometry}

\usepackage{amsmath}
\usepackage{amsfonts}
\usepackage{graphicx}
\usepackage{amssymb}
\usepackage{algorithm}
\usepackage{algpseudocode}
\usepackage{amsthm}
\usepackage[colorlinks=true, allcolors=blue]{hyperref}
\usepackage{xcolor}

\theoremstyle{definition}
\newtheorem{definition}{Definition}

\newtheorem{theorem}{Theorem}[section]
\newtheorem{lemma}[theorem]{Lemma}
\newtheorem{corollary}[theorem]{Corollary}
\newtheorem{remark}{Remark}

\newcommand{\E}{\mathbb{E}}
\renewcommand{\P}{\mathbb{P}}

\newcommand{\pG}{\mathcal{\psi}}

\newcommand{\rH}{\mathcal{H}}

\title{Achievability of Heterogeneous Hypergraph Recovery from its Graph Projection}
\author{Alexander Morgan \thanks{xmorgan@mit.edu} \and Chenghao Guo \thanks{chenghao@mit.edu}}
\date{}

\begin{document}
\maketitle

\begin{abstract}
  We formulate and analyze a heterogeneous random hypergraph model, and we provide an achieveability result for recovery of hyperedges from the observed projected graph. 
  We observe a projected graph which combines random hyperedges across all degrees, where a projected edge appears if and only if both vertices appear in at least one hyperedge. Our goal is to reconstruct the original set of hyperedges of degree $d_j$ for some $j$. 
  Our achievability result is based on the idea of selecting maximal cliques of size $d_j$ in the projected graph, and 
  we show that this algorithm succeeds under a natural condition on the densities.
  This achievability condition generalizes a known threshold for $d$-uniform hypergraphs with noiseless and noisy projections. We conjecture the threshold to be optimal for recovering hyperedges with the largest degree. 
\end{abstract}

\section{Introduction}

The problem of hyperedge recovery naturally arises as part of a larger class of problems in community detection. Common community detection problems involve the study of pairwise relations in settings such as friendship connections and social networks, where connections are formed between any pair of vertices sharing a common latent cluster. Inspired by real-world social networks, significant attention has been given to graphs with group structures, such as random intersection graphs \cite{BBN20, vdHKV21}. Additionally, researchers have extensively studied the hypergraph stochastic block model \cite{ACKZ15, DWar, GD17, GJ23, GP23, GP24, SZ22}.

This problem has applications in real-world settings like scientific co-authorship networks \cite{New04}. The theoretical study was initiated in recent works \cite{reconstruct, bresler2025partial} which determined the sharp threshold for reconstructing a hypergraph from its graph projection in a noiseless setting. Beyond theoretical bounds, algorithmic methods have been developed to address the problem of reconstruction from projected graphs, which has been extensively studied both theoretically and empirically in recent years. These include Bayesian approaches to sample from the posterior distribution given the projected graph \cite{LYA23, YPP21} , scoring methods based on sampled hypergraphs \cite{WK24} , and foundational models to recover weighted hypergraphs \cite{CFLL24}. A comprehensive overview of this model and its various applications in real networks can be found in \cite{bresler2025partial}.

However, the setting discussed in \cite{reconstruct,bresler2025partial} is far from realistic. Real networks have noise, and more importantly, have \emph{interactions of different sizes}. The noisy setting is discussed in \cite{gong2025detection}, which can be viewed as a mixture of degree-2 and degree-$d$ hyperedges. In their work, the observed graph is a noisy projection where true edges are kept with probability $p$ and non-edges are added with probability $q$. Their results established sharp thresholds for both detection and reconstruction, notably revealing a detection-reconstruction gap phenomenon in this problem. 

General heterogeneous hypergraphs is a more challenging problem to study. Here, we analyze the problem of recovery of a heterogeneous hypergraph from its projected graph. In particular, we study the following model.
\begin{definition}[Heterogeneous Random Hypergraph]\label{def:heterogeneous_random_hypergraph}
    Consider a set of vertices $V = [n]$, and fix hyperedge degrees $2 \leq d_1 < \dots < d_k$ with corresponding $\delta_1, \dots, \delta_k$, where $0 < \delta_j < 1$. Assume each hyperedge of degree $d_j$ appears with identical probability $p_j = n^{1 - d_j + \delta_j + o(1)}$ and that all hyperedges across all degrees appear independently. An edge appears in the projected graph if and only if both vertices in the edge appear together in at least one hyperedge.
\end{definition}

Having observed the projected graph due to hyperedges across all degrees, we attempt to recover hyperedges of degree $d_j$. Suppose $\rH_j$ is the true set of hyperedges of degree $d_j$, and $\widehat{\rH}_j$ is the set we estimate. Then, we consider recovery successful if the symmetric difference between these two sets has size that is $o(1)$ with respect to the expected size of $\rH_j$. In particular, our goal is to achieve
\begin{equation*}
    \E\left[\left|\rH_j \triangle \widehat{\rH}_j\right|\right] = o(1) \E\left[\left|\rH_j\right|\right],
\end{equation*}
which we use as our criterion for successful recovery.

\subsection{Results}
We analyze the performance of the following algorithm. Here, \emph{maximal clique} refers to a clique that is not a subset of a larger clique.
\begin{algorithm}[H]
\caption{Maximal Clique Estimation}\label{alg:maximal_clique_estimation}
\begin{algorithmic}
\State Estimate the hyperedges of degree $d_j$ as the set of all maximal cliques of size $d_j$.
\end{algorithmic}
\end{algorithm}

Under the heterogeneous hypergraph model as stated in Definition \ref{def:heterogeneous_random_hypergraph}, we will show that recovery of hyperedges of degree $d_j \geq 3$, i.e., recovery of $1 - o(1)$ hyperedges in expectation, is possible if
\begin{equation}\label{eqn:key_inequality}
    \delta^* \triangleq \max\{\delta_1, \dots, \delta_k\} < \frac{d_j - 2}{d_j} + \frac{2 \delta_j}{d_j(d_j - 1)}
\end{equation}
by applying Algorithm \ref{alg:maximal_clique_estimation}. In particular, we'll show the following result.

\begin{theorem}\label{theorem:maximal_clique_recovery}
    Consider the model specified in Definiton \ref{def:heterogeneous_random_hypergraph}. Fix a $j$ with $d_j \ge 3$, and suppose Equation \ref{eqn:key_inequality} holds for this $j$. Let $\rH_j$ be the set of true hyperedges of degree $d_j$. Estimating $\rH_j$ by applying Algorithm \ref{alg:maximal_clique_estimation} to the projected graph $\pG$ yielding $\widehat{\rH}_j$ succeeds, i.e., 
    \begin{equation*}
        \E\left[\left|\rH_j \triangle \widehat{\rH}_j\right|\right] = o(1) \E\left[\left|\rH_j\right|\right].
    \end{equation*}
\end{theorem}

Cliques of size $d_j$ can appear due to a true hyperedge, but they can also appear due to hyperedges of various degrees combining to form a cover of that clique in the projected hypergraph. One way this can occur is if all $\binom{d_j}{2}$ edges are individually covered due to ``noise'' hyperedges. Other configurations are also possible. However, we will show that if we can guarantee that the density of cliques of size $d_j$ due to the all-edges noise configuration is dominated by the density of true hyperedges of size $d_j$, i.e., 
\begin{equation}\label{eqn:key_inequality_density_interpretation}
    -\binom{d_j}{2} + \binom{d_j}{2} \delta^* < 1 - d_j + \delta_j,
\end{equation}
where
\begin{equation*}
    \delta^* \triangleq \max\{\delta_1, \dots, \delta_k\},
\end{equation*}
then Algorithm \ref{alg:maximal_clique_estimation} succeeds. In particular, the all-edges noise configuration is the most detrimental configuration we need to consider.

Note that Equation \ref{eqn:key_inequality} is equivalent to Equation \ref{eqn:key_inequality_density_interpretation}.
Simultaneous recovery of all hyperedges $d_j \geq 3$ is possible if all required inequalities in Equation \ref{eqn:key_inequality} are simultaneously satisfied across $j$. By extracting the requirement from Equation \ref{eqn:key_inequality} that
\begin{equation*}
    \delta_j < \frac{d_j - 2}{d_j} + \frac{2 \delta_j}{d_j(d_j - 1)},
\end{equation*}
which is equivalent to the condition
\begin{equation*}
    \delta_j < \frac{d - 1}{d + 1},
\end{equation*}
we see that our achievability result naturally generalizes the $\frac{d - 1}{d + 1}$ threshold in the $d$-uniform hypergraph setting. But we emphasize that this is an achievability result, and better performance may be possible for recovery of hyperedges of degree $d_j < d_k$, for example.

In the $d$-uniform hypergraph setting, the two dominating configurations leading to cliques of size $d$ are the true hyperedge configuration and the all-edges noise configuration. These competing configurations lead to the $\frac{d - 1}{d + 1}$ threshold. Theorem \ref{theorem:maximal_clique_recovery} states that achievability extends for a generalized version of this idea, where noise may now occur due to hyperedges of different degrees.

\section{Notation}
We use $d$ to denote the degree of a hyperedge. We use $p$ to denote the probability a particular hyperedge appears, where $p = n^{1 - d + \delta + o(1)}$ with $0 < \delta < 1$ as a fixed parameter tied to $d$. 

In the heterogeneous random hypergraph model, we consider $k$ such degrees, i.e., we let $2 \leq d_1 < \cdots < d_k$, where hyperedges of degree $d_j$ appear independently with identical probability $p_j = n^{1 - d_j + \delta_j + o(1)}$, with $0 < \delta_j < 1$. Hyperedges of different degrees are assumed to appear independently as well.

We use $\rH$ to denote a random and possibly heterogeneous hypergraph, and $\rH_j$ to denote the random hypergraph consisting of only hyperedges of degree $d_j$ under the heterogeneous model. We observe $\pG$, which is the random projected hypergraph due to $\rH$.

We use $E$ for edge sets and $V$ for vertex sets. Also, for convenience, $[d] = \{1, \dots, d\}$. Given a set of edges $E$, we use $\mathcal{U}$ to denote a cover for these edges. In particular, let $V$ denote the set of vertices spanned by $E$. Then, $\mathcal{U} \subset 2^{V}$ satisfies
\begin{itemize}
    \item
    For all $u \in \mathcal{U}$, $|u| \geq 2$, and $E \not\subset \bigcup_{u' \in \mathcal{U}\setminus\{u\}} \binom{u'}{2}$
    \item
    $E \subset \bigcup_{u \in \mathcal{U}} \binom{u}{2}$.
\end{itemize}
Finally, $\text{Bin}(n, p)$ denotes the binomial distribution, and we will use $A$ to denote a probabilistic event, where the specific event being considered will made clear from context.

\section{Preliminiaries}
We will need the following basic probability facts. Proofs are left to the appendix.
\begin{lemma}\label{lemma:binom_union_bound_tight_general}
    Let $X_n \sim \text{Bin}\left(m_n, p_n\right)$ with $m_n \to +\infty$ and $m_n p_n \to 0$. Then,
    \begin{equation}\label{eqn:binom_union_bound_tight_general_1}
        \P(X_n \geq 1) = (1 - o(1)) m_n p_n.
    \end{equation}
    More specifically, 
    \begin{equation}\label{eqn:binom_union_bound_tight_general_2}
        \P(X_n \geq 1) \geq (1 - o(1)) m_n p_n
    \end{equation}
    and
    \begin{equation}\label{eqn:binom_union_bound_tight_general_3}
        \P(X_n \geq 1) \leq m_n p_n.
    \end{equation}
\end{lemma}

\begin{corollary}\label{corollary:binom_union_bound_tight}
    Let $X_n \sim \text{Bin}\left(n^{d - d' + o(1)}, n^{1 - d + \delta + o(1)}\right)$ with $\delta < 1$ and $2 \leq d' < d$. Then,
    \begin{equation*}
        \P(X_n \geq 1) = n^{1 - d' + \delta + o(1)}.
    \end{equation*}
    While not assumed here, we will later apply this for integer $d', d$.
\end{corollary}
\begin{lemma}[Max Rate Dominates]\label{lemma:max_rate_dominates}
    Fix $r \geq 1$ and suppose $A_1, \dots, A_r$ are events with $\P(A_j) = n^{\alpha_j + o(1)}$ and $\alpha_j < 0$. Then,
    \begin{equation*}
        \P\left(\bigcup_{j = 1}^r A_j\right) = n^{\max_{j}{\alpha_j} + o(1)}.
    \end{equation*}
    Note that we don't require independence for this result.
\end{lemma}

\section{Outline}
Next, we formally prove Theorem~\ref{theorem:maximal_clique_recovery}. In Section~\ref{sec:prob-events}, we prove the probability of subgraphs appearing in the projected graph $\pG$ is controlled by an optimization problem $g$. In Section~\ref{sec:g}, we solve this optimization problem in core cases that we care about. The proof is concluded in Section~\ref{sec:proof-achievability}.

\section{Probabilities of Events on $\pG$}\label{sec:prob-events}
The goal of this section is to characterize the probability a certain subgraph $E$ appears in the projected graph, $\P(E\subset \pG)$.

The utility of Lemma \ref{lemma:binom_union_bound_tight_general} and Corollary \ref{corollary:binom_union_bound_tight} is that they allow us to infer the probability that cliques of smaller size $d'$ appear due to a single hyperedge of degree $d$. Lemma \ref{lemma:implied_hyperedge_probability} formalizes this.
\begin{lemma}[Implied Hyperedge Probability]\label{lemma:implied_hyperedge_probability}
    Suppose hyperedges of degree $d \geq 2$ appear i.i.d. with probability $p = n^{1 - d + \delta + o(1)}$, with $0 < \delta < 1$. Consider fixed (with respect to $n$) vertex sets $V, V'$, with $V' \subset V$ and $2 \leq d' = |V'| \leq d$. Let $A_{V'}$ denote the event that there exists a hyperedge that contains all nodes in $V'$ and none in $V \setminus V'$. Then,
    \begin{equation}\label{eqn:implied_hyperedge_probability_lemma}
        \P(A_{V'}) = n^{1 - d' + \delta + o(1)}.
    \end{equation}
    More specifically,
    \begin{align*}
        \P(A_{V'}) &= \P\left(\text{Bin}\left(\binom{n - |V|}{d - d'}, p\right) \geq 1\right)\\
        &= 1 - (1 - p)^{\binom{n - |V|}{d - d'}}\\
        &= (1 - o(1)) \binom{n - |V|}{d - d'} p,
    \end{align*}
    with the upper bound
    \begin{equation*}
        \P(A_{V'}) \leq \binom{n - |V|}{d - d'} p.
    \end{equation*}
    For other candidate vertex sets $V'' \subset V$ with $V'' \neq V'$ and $2 \le |V''| \le d$, $A_{V'}$ and $A_{V''}$ are independent since they depend on disjoint sets of hyperedges, hence the interpretation ``implied hyperedge probability.''
\end{lemma}
\begin{proof}
    The case $d' = d$ is immediate. When $d' < d$, there are $\binom{n - |V|}{d - d'} = n^{d - d' + o(1)}$ candidate hyperedges contributing to $A_{V'}$, each independently appearing with identical probability $p = n^{1 - d + \delta + o(1)}$. The results follow from Lemma \ref{lemma:binom_union_bound_tight_general} and Corollary \ref{corollary:binom_union_bound_tight}.
\end{proof}
Next, we work toward the case of multiple degrees $d_1 < \dots < d_k$. This allows us to now derive a slightly more general version of Lemma \ref{lemma:implied_hyperedge_probability}.

\begin{lemma}[General Implied Hyperedge Probability]\label{lemma:general_implied_hyperedge_probability}
    Consider the model specified by Definition \ref{def:heterogeneous_random_hypergraph}. Again fix vertex sets $V, V'$, with $V' \subset V$ and $2 \leq d' = |V'| \leq d_k$. Let $A_{V'}$ denote the event that there exists a hyperedge (of any degree) that contains all nodes in $V'$ and none in $V \setminus V'$. Then,
    \begin{equation}\label{eqn:general_implied_hyperedge_probability_lemma}
        \P(A_{V'}) = n^{1 - d' + \Delta + o(1)}
    \end{equation}
    where
    \begin{equation*}
        \Delta \triangleq \max_{j : d_j \geq d'} \delta_j.
    \end{equation*}
    More specifically,
    \begin{align*}
        \P(A_{V'}) &= 1 - \prod_{j : d_j \geq d'} (1 - p_j)^{\binom{n - |V|}{d_j - d'}}\\
        &= (1 - o(1)) \sum_{j : d_j \geq d'} \binom{n - |V|}{d_j - d'} p_j,
    \end{align*}
    with the upper bound
    \begin{equation*}
        \P(A_{V'}) \leq \sum_{j : d_j \geq d'}\binom{n - |V|}{d_j - d'} p_j.
    \end{equation*}
    The comment in Lemma \ref{lemma:implied_hyperedge_probability} regarding other candidate vertex sets $V'' \subset V$ still applies since all hyperedges appear independently.
\end{lemma}
\begin{proof}
    For any $j$ such that $d_j \geq d'$, let $\mathcal{A}_{j, V'}$ denote the event that there exists a hyperedge of degree $d_j$ that contains all nodes in $V'$ and none in $V \setminus V'$. Then, by Lemmas \ref{lemma:implied_hyperedge_probability} and \ref{lemma:max_rate_dominates},
    \begin{align*}
        \P(A_{V'}) &= \P\left(\bigcup_{j : d_j \geq d'}\mathcal{A}_{j, V'}\right)\\
        &= n^{1 - d' + \Delta + o(1)}.
    \end{align*}
    By the union bound,
    \begin{equation*}
        \P(A_{V'}) \leq \sum_{j : d_j \geq d'}\binom{n - |V|}{d_j - d'} p_j,
    \end{equation*}
    and Bonferroni's inequality implies
    \begin{align*}
        \P(A_{V'}) &\geq \sum_{j : d_j \geq d'} (1 - o(1)) \binom{n - |V|}{d_j - d'} p_j \left[1 - \sum_{r > j} (1 - o(1)) \binom{n - |V|}{d_r - d'} p_r\right]\\
        &\geq (1 - o(1))\sum_{j : d_j \geq d'}\binom{n - |V|}{d_j - d'} p_j.
    \end{align*}
\end{proof}

We can now analyze the probability specific edge configurations appear in the projected graph.

\begin{theorem}\label{theorem:clique_density}
    Consider the model specified by Definition \ref{def:heterogeneous_random_hypergraph}. Fix a set of edges $E$, and let $V$ be the set of vertices spanned by these edges. Then, with $\pG$ as the combined projected graph due to hyperedges of all degrees,
    \begin{equation*}
        \P\left(E\subset \pG
        \right) = n^{g(E, \mathbf{\Delta}) + o(1)},
    \end{equation*}
    where $\mathbf{\Delta} = (\mathbf{\Delta}_2, \dots, \mathbf{\Delta}_{|V|})$ is set by $\mathbf{\Delta}_j = \max_{r: d_r \geq j} \delta_r$ with the convention that $\mathbf{\Delta}_j = -\infty$ when $j > d_k$, and $g(E, \mathbf{\Delta})$ we now define.
    For a choice of $\mathbf{\Delta}
    = (\mathbf{\Delta}_2, \dots, \mathbf{\Delta}_{|V|})$ with $\mathbf{\Delta}_j < 1$, let
    \begin{equation}
        g(E, \mathbf{\Delta}) \triangleq \max_{\mathcal{U}} \sum_{u \in \mathcal{U}} 1 + \mathbf{\Delta}_{|u|} - |u|,
    \end{equation}
    where the maximum is taken over all valid $\mathcal{U}$ covering $E$, i.e., we require $\mathcal{U} \subset 2^V$ to satisfy
    \begin{itemize}
        \item
        For all $u \in \mathcal{U}$, $|u| \geq 2$ and $E \not\subset \bigcup_{u' \in \mathcal{U}\setminus\{u\}} \binom{u'}{2}$
        \item
        $E \subset \bigcup_{u \in \mathcal{U}} \binom{u}{2}$.
    \end{itemize}
    Given a fixed $E$, there are finitely many candidate coverings $\mathcal{U}$. 
    
    For better control of the probability, we have the lower bound
    \begin{equation*}
        \P(E\subset \pG) \geq (1 - o(1))\prod_{u \in \mathcal{U}} \sum_{j : d_j \geq |u|} \binom{n - |V|}{d_j - |u|} p_j
    \end{equation*}
    for any particular cover $\mathcal{U}$, and the upper bound
    \begin{equation*}
        \P(E\subset \pG) \leq \sum_{\mathcal{U}}\prod_{u \in \mathcal{U}} \sum_{j : d_j \geq |u|} \binom{n - |V|}{d_j - |u|} p_j.
    \end{equation*}
\end{theorem}
\begin{proof}
    Fix a candidate cover $\mathcal{U}$ such that $|u| \leq d_k$ for all $u \in \mathcal{U}$. For a particular $u \in \mathcal{U}$, let $A_{u}$ denote the event that there exists a hyperedge that contains all vertices in $u$ but none in $V \setminus u$, and let $A_{\mathcal{U}} = \bigcap_{u \in \mathcal{U}} A_u$. Then, Lemma \ref{lemma:general_implied_hyperedge_probability} implies that
    \begin{equation*}
        \P(A_u) = n^{1 - |u| + \mathbf{\Delta}_{|u|} + o(1)},
    \end{equation*}
    so by independence,
    \begin{equation*}
        \P(A_{\mathcal{U}}) = n^{\sum_{u \in \mathcal{U}} 1 - |u| + \mathbf{\Delta}_{|u|} + o(1)}.
    \end{equation*}
    Let $A$ denote the event that all edges $E$ appear in the projected graph, i.e., $A = \{E \subset \psi\}$. Note that $A = \cup_{\mathcal{U}} A_{\mathcal{U}}$. It's clear that $A_{\mathcal{U}} \subset A$ since any random sample in $A_{\mathcal{U}}$ has a projected graph containing all edges in $E$, so $\cup_{\mathcal{U}} A_{\mathcal{U}} \subset A$. Also, for any random sample in $A$, consider all hyperedges that contain at least two nodes in $V$. Greedily drop (in any order) hyperedges that are redundant for covering $E$. Then, the remaining intersections between hyperedges and $V$ represent a valid cover, so $A \subset \cup_{\mathcal{U}} A_{\mathcal{U}}$.
    Finally, since there are finitely many candidate covers $\mathcal{U}$, Lemma \ref{lemma:max_rate_dominates} implies
    \begin{equation*}
        \P(A) = n^{g(E, \mathbf{\Delta}) + o(1)}.
    \end{equation*}
    The lower and upper bounds on $\P(A)$ also follow from similar analysis and Lemma \ref{lemma:general_implied_hyperedge_probability}.
\end{proof}
\begin{remark}
    
It will be convenient to introduce a restricted version of this function
\begin{equation}
    g(E, \mathbf{\Delta}; M) \triangleq \max_{\mathcal{U} : |\mathcal{U}| = M} \sum_{u \in \mathcal{U}} 1 + \mathbf{\Delta}_{|u|} - |u|
\end{equation}
such that
\begin{equation}
    g(E, \mathbf{\Delta}) = \max_{1 \leq M \leq \binom{|V|}{2}} g(E, \mathbf{\Delta}; M).
\end{equation}
It can also be convenient use $\delta^* = \max\{\delta_1, \dots, \delta_k\}$ to bound
\begin{equation*}
    g(E, \mathbf{\Delta}) \leq g(E, \delta^*) \triangleq \max_{\mathcal{U}}\sum_{u \in \mathcal{U}} 1 + \delta^* - |u|,
\end{equation*}
with the maximum taken over covers $\mathcal{U}$ satisfying the additional constraint that $|u| \leq d_k$ for all $u \in \mathcal{U}$.
\end{remark}

\section{Combinatorial Properties of $g(E, \mathbf{\Delta})$}\label{sec:g}
Now, we establish two key Lemmas regarding $g(E, \mathbf{\Delta})$, with proofs deferred to the appendix.
\begin{lemma}\label{lemma:optimal_cover_without_entire_clique}
    Fix vertices $V = \{1, \dots, d\}$ with $d \geq 3$, and fix $0 < \delta < 1$. Let $\mathbf{\Delta}_j = \delta$ for $2 \leq j \leq d - 1$ and $\mathbf{\Delta}_d = -\infty$. Thus, we have implicitly constrained sets in $\mathcal{U}$ to have size at most $d - 1$. 
    Then,
    \[
    g\left(\binom{V}{2}, \mathbf{\Delta}\right) = \max\left\{\delta d - 2d + 3, -\binom{d}{2} + \binom{d}{2}\delta\right\},
    \]
    achieved by the ``$d - 1$ clique plus star'' cover $\mathcal{U} = \{\{1, \dots, d - 1\}\} \cup \{\{1, d\}, \{2, d\}, \dots, \{d - 1, d\}\}$ or the individual-edges cover $\mathcal{U} = \binom{V}{2}$.
\end{lemma}

\begin{lemma}\label{lemma:optimal_star_cover}
    Fix vertices $V = \{1, \dots, d\}$ with $d \geq 2$, and suppose $\mathbf{\Delta}$ satisfies $0 < \mathbf{\Delta}_j < 1$ for $2 \leq j \leq d$ and
    \begin{equation*}
        \delta \triangleq \mathbf{\Delta}_2 = \max_{2 \leq j \leq d} \mathbf{\Delta}_j.
    \end{equation*}
    Define the edge set $E = \{(1, j)\}_{j = 2}^d$, representing a star graph configuration. Then, $g(E; \mathbf{\Delta}) = -(d - 1) + (d - 1) \delta$ achieved by the individual edges cover $\mathcal{U} = E = \{\{1, j\}\}_{j = 2}^d$.
\end{lemma}

\section{Proof of the Achievability Result, Theorem~\ref{theorem:maximal_clique_recovery}}\label{sec:proof-achievability}
\begin{proof}
    We have that
    \begin{equation*}
        \mathbb{E}\left[\left|\rH_j\right|\right] = \binom{n}{d_j} n^{1 - d_j + \delta_j + o(1)}
    \end{equation*}\ and
    \begin{equation*}
        \E\left[\left|\rH_j \triangle \widehat{\rH}_j\right|\right] = \binom{n}{d_j} \left[\P(A_{\mathrm{FP}}) + \P(A_{\mathrm{FN}})\right],
    \end{equation*}
    where $A_{\mathrm{FP}}$ is the false positive event for estimating $[d_j]$, i.e., $A_{\mathrm{FP}} = \{[d_j] \notin \rH_j\} \cap \{[d_j] \in \widehat{\rH}_j\}$, and $A_{\mathrm{FN}}$ is the false negative event for estimating $[d_j]$, i.e., $A_{\mathrm{FN}} = \{[d_j] \in \rH_j\} \cap \{[d_j] \notin \widehat{\rH}_j\}$. Thus, it's sufficient to show that
    \begin{equation*}
        \P(A_{\mathrm{FP}}) = o(1) \cdot n^{1 - d_j + \delta_j + o(1)}
    \end{equation*}
    and
    \begin{equation*}
        \P(A_{\mathrm{FN}}) = o(1) \cdot n^{1 - d_j + \delta_j + o(1)}.
    \end{equation*}

    First, we analyze the rate of false positives. Note that any hyperedge of size $d_r \geq d_j$ that contains all nodes $[d_j]$ will not contribute to the false positive probability since that would cause $[d_j]$ to not be maximal, or it would mean that $[d_j]$ is a real hyperedge. Thus, the only way for a false positive to be achieved is if a cover $\mathcal{U}$ for $E = \binom{[d_j]}{2}$ with $|\mathcal{U}| \geq 2$ is formed. The probability that $[d_j]$ is not a hyperedge is $1 - o(1)$, and by a minor modification to the proof of Theorem \ref{theorem:clique_density}, the probability that all edges appear in $E$ due to a cover of size $|\mathcal{U}| \geq 2$ is $n^{g(E; \tilde{\mathbf{\Delta}}) + o(1)}$ with $\tilde{\mathbf{\Delta}}_\ell = \max_{r: d_r \geq \ell} \delta_r$ for $\ell < d_j$ and $\tilde{\mathbf{\Delta}}_{d_j} = -\infty$. These events are independent, so
    \begin{equation*}
        \P(A_{\mathrm{FP}}) \le n^{g(E; \tilde{\mathbf{\Delta}}) + o(1)}.
    \end{equation*}
    Now, let $\mathbf{\Delta}_\ell = \delta^*$ for $\ell < d_j$ and $\mathbf{\Delta}_{d_j} = -\infty$. Then, by applying Lemma \ref{lemma:optimal_cover_without_entire_clique},
    \begin{align*}
        \P(A_{\mathrm{FP}}) &\leq n^{g(E; \mathbf{\Delta}) + o(1)}\\
        &= n^{\max\left\{-\binom{d_j}{2} + \binom{d_j}{2} \delta^*, \delta^* d_j - 2 d_j + 3\right\} + o(1)}.
    \end{align*}
    Now, Equation \ref{eqn:key_inequality} is equivalent to
    \begin{equation*}
        -\binom{d_j}{2} + \binom{d_j}{2} \delta^* < 1 - d_j + \delta_j.
    \end{equation*}
    Also,
    \begin{align*}
        \delta^* d_j - 2 d_j + 3 &< 1 - d_j + \frac{2\delta_j}{d_j - 1}\\
        &\leq 1 - d_j + \delta_j
    \end{align*}
    for $d_j \geq 3$. Thus, 
    \begin{equation*}
        \P(A_{\mathrm{FP}}) = o(1) \cdot n^{1 - d_j + \delta_j + o(1)}.
    \end{equation*}

    Now, we analyze the probability of a false negative for $[d_j]$. For a false negative to occur, the hyperedge $[d_j]$ must exist, and all nodes in $[d_j]$ must be connected to an outside node, i.e., we must observe any star formation connecting $[d_j]$ to an external node. These two events are independent. The probability $[d_j]$ is a hyperedge is $p_j = n^{1 - d_j + \delta_j + o(1)}$, and by Theorem \ref{theorem:clique_density} and Lemma \ref{lemma:optimal_star_cover}, the probability we observe a star formation connecting $[d_j]$ to a particular outside node is $n^{-d_j + d_j\delta^* + o(1)}$.
    The probability we observe a star formation connecting $[d_j]$ to any outside node is, by the union bound, at most $n^{1 - d_j + d_j\delta^* + o(1)}$.
    Then, by independence,
    \begin{align*}
        \P(A_{\mathrm{FN}}) &\leq n^{1 - d_j + \delta_j + o(1)} \cdot n^{1 - d_j + d_j\delta^* + o(1)}.
    \end{align*}
    Finally, by Equation \ref{eqn:key_inequality},
    \begin{equation*}
        1 - d_j + d_j\delta^* < -1 + \frac{2\delta_j}{d_j - 1} \leq 0
    \end{equation*}
    for $d_j \geq 3$, so
    \begin{align*}
        \P(A_{\mathrm{FN}}) &\leq o(1)\cdot n^{1 - d_j + \delta_j + o(1)},
    \end{align*}
    completing the proof.
\end{proof}

\section{Conclusion}
We demonstrate that in the heterogeneous hypergraph model, the maximal clique estimator defined by Algorithm \ref{alg:maximal_clique_estimation} recovers hyperedges of degree $d_j$ with $o(1)$ symmetric difference error if the density of the all-edges configuration for cliques of size $d_j$ is dominated by the density of true hyperedges of degree $d_j$. This achieveability result generalizes prior thresholds for recovery in $d$-uniform hypergraphs in a natural way, providing a conceptually simple approach to community detection in this setting.

Here we also list some open problems.
\begin{enumerate}
    \item Converse of achievability. We conjecture that the threshold in Theorem~\ref{theorem:maximal_clique_recovery} is optimal for the recovery of hyperedges with the largest degree. 
    \item The general optimal recovery threshold for heterogeneous hypergraphs. The threshold in Theorem~\ref{theorem:maximal_clique_recovery} is likely not optimal when $\delta_1,\cdots, \delta_k$ is arbitrary and the objective is to recover hyperedges of all different degrees. 
    \item Exact recovery was discussed in \cite{reconstruct} and \cite{bresler2025partial} for $d$-uniform hypergraphs. However, exact recovery for heterogeneous random hypergraphs is not understood, even in the noisy setting where there are only degree-2 and degree-$d$ hyperedges present.
\end{enumerate}

\bibliographystyle{alpha}
\bibliography{ref}

@inproceedings{reconstruct,
  title = 	 "{Thresholds for Reconstruction of Random Hypergraphs From Graph Projections}",
  author =       {Bresler, Guy and Guo, Chenghao and Polyanskiy, Yury},
  booktitle = 	 {Proceedings of Thirty Seventh Conference on Learning Theory},
  pages = 	 {632--647},
  year = 	 {2024},
  volume = 	 {247},
  series = 	 {Proceedings of Machine Learning Research},
  publisher =    {PMLR},
}

@article{bresler2025partial,
  title={Partial and Exact Recovery of a Random Hypergraph from its Graph Projection},
  author={Bresler, Guy and Guo, Chenghao and Polyanskiy, Yury and Yao, Andrew},
  journal={arXiv preprint arXiv:2502.14988},
  year={2025}
}

@article{gong2025detection,
  title={Detection and Reconstruction of a Random Hypergraph from Noisy Graph Projection},
  author={Gong, Shuyang and Li, Zhangsong and Xu, Qiheng},
  journal={arXiv preprint arXiv:2506.17527},
  year={2025}
}

@inproceedings{ACKZ15,
  title={Spectral detection on sparse hypergraphs},
  author={Angelini, Maria C. and Caltagirone, Francesco and Krzakala, Florent and Zdeborov{\'a}, Lenka},
  booktitle={Proceedings of the 53rd Annual Allerton Conference on Communication, Control, and Computing (Allerton)},
  pages={66--73},
  year={2015},
  organization={IEEE}
}

@article{BBN20,
  title={Phase transitions for detecting latent geometry in random graphs},
  author={Brennan, Matthew and Bresler, Guy and Nagaraj, Dheeraj M.},
  journal={Probability Theory and Related Fields},
  volume={178},
  number={3-4},
  pages={1215--1289},
  year={2020}
}

@inproceedings{CFLL24,
  title={Relational learning in pretrained models: A theory from hypergraph recovery perspective},
  author={Chen, Yang and Fang, Cong and Lin, Zhouchen and Liu, Bing},
  booktitle={Proceedings of the 41st International Conference on Machine Learning (ICML)},
  pages={6666--6698},
  year={2024},
  organization={PMLR}
}

@article{DWar,
  title={Optimal and exact recovery on general non-uniform hypergraph stochastic block model},
  author={Dumitriu, Ioana and Wang, Hai-Xiao},
  journal={Annals of Statistics},
  note={to appear}
}

@article{GD17,
  title={Consistency of spectral partitioning of uniform hypergraphs under planted partition model},
  author={Ghoshdastidar, Debarghya and Dukkipati, Ambedkar},
  journal={Annals of Statistics},
  volume={45},
  number={1},
  pages={289--315},
  year={2017}
}

@inproceedings{GJ23,
  title={Community detection in the hypergraph SBM: Exact recovery given the similarity matrix},
  author={Gaudio, Julia and Joshi, Nirmit},
  booktitle={Proceedings of 36th Conference on Learning Theory (COLT)},
  pages={469--510},
  year={2023},
  organization={PMLR}
}

@inproceedings{GP23,
  title={Weak recovery threshold for the hypergraph stochastic block model},
  author={Gu, Yuzhou and Polyansky, Yury},
  booktitle={Proceedings of the 36th Annual Conference on Learning Theory (COLT)},
  pages={885--920},
  year={2023},
  organization={PMLR}
}

@inproceedings{GP24,
  title={Community detection in the hypergraph stochastic block model and reconstruction on hypertrees},
  author={Gu, Yuzhou and Polyansky, Yury},
  booktitle={Proceedings of the 37th Annual Conference on Learning Theory (COLT)},
  pages={2166--2203},
  year={2024},
  organization={PMLR}
}

@article{LYA23,
  title={Hypergraph reconstruction from uncertain pairwise observations},
  author={Lizotte, Simon and Young, Jean-Gabriel and Allard, Antoine},
  journal={Scientific Reports},
  volume={13},
  number={1},
  pages={21364},
  year={2023}
}

@article{New04,
  title={Coauthorship networks and patterns of scientific collaboration},
  author={Newman, Mark E. J.},
  journal={Proceedings of the National Academy of Sciences of the United States of America},
  volume={101},
  number={1},
  pages={5200--5205},
  year={2004}
}

@inproceedings{SZ22,
  title={Sparse random hypergraphs: Non-backtracking spectra and community detection},
  author={Stephan, Ludovic and Zhu, Yizhe},
  booktitle={IEEE 63rd Annual Symposium on Foundations of Computer Science (FOCS)},
  pages={567--575},
  year={2022},
  organization={IEEE}
}

@article{vdHKV21,
  title={Random intersection graphs with communities},
  author={van der Hofstad, Remco and Komj{\'a}thy, J{\'u}lia and Vadon, Vikt{\'o}ria},
  journal={Advances in Applied Probability},
  volume={53},
  number={4},
  pages={1061--1089},
  year={2021}
}

@inproceedings{WK24,
  title={From graphs to hypergraphs: Hypergraph projection and its reconstruction},
  author={Wang, Yanbang and Kleinberg, Jon},
  booktitle={Proceedings of the 12th International Conference on Learning Representations (ICLR)},
  year={2024}
}

@article{YPP21,
  title={Hypergraph reconstruction from network data},
  author={Young, Jean-Gabriel and Petri, Giovanni and Peixoto, Tiago P.},
  journal={Communications Physics},
  volume={4},
  number={1},
  pages={135},
  year={2021}
}

\appendix

\section{Postponed Proofs}
\subsection{Proof of Lemma~\ref{lemma:binom_union_bound_tight_general}}
\begin{proof}
    We have
    \begin{align*}
        \P(X_n \geq 1) &\geq \P(X_n = 1)\\
        &\geq m_n p_n \left(1 - p_n\right)^{m_n}\\
        &= m_n p_n \left(1 - \frac{m_n p_n}{m_n}\right)^{m_n}\\
        &= m_n p_n (1 - o(1)).
    \end{align*}
    Also, the union bound or Markov's inequality yields Equation \ref{eqn:binom_union_bound_tight_general_3}.
\end{proof}

\subsection{Proof of Lemma~\ref{lemma:max_rate_dominates}}
\begin{proof}
    \begin{align*}
        \P\left(\bigcup_{j = 1}^r A_j\right) &\geq \max_{j} \P(A_j)\\
        &\geq n^{\max_{j}{\alpha_j} + o(1)}
    \end{align*}
    and
    \begin{align*}
        \P\left(\bigcup_{j = 1}^r A_j\right) &\leq r \left[\max_{j} \P(A_j)\right]\\
        &\leq n^{\max_{j}{\alpha_j} + o(1)}.
    \end{align*}
\end{proof}

\subsection{Proof of Lemma~\ref{lemma:optimal_cover_without_entire_clique}}
\begin{proof}
    Here, we generalize a Lemma from \cite{reconstruct}, and the proof is similar. We include the reasoning here for completeness, and we will slightly simplify the argument. 

    We can write the maximization as
    \begin{equation*}
        g(E, \mathbf{\Delta}) = \max_{2 \leq M \leq \binom{d}{2}} \max_{\mathcal{U} : |\mathcal{U}| = M} \sum_{u \in \mathcal{U}} 1 + \delta - |u|,
    \end{equation*}
    where we restrict $M \geq 2$ because we can disregard the case $\mathcal{U} = \{V\}$. The inner maximization is taken over all valid covers $\mathcal{U}$ satisfying the usual conditions.
    Now, set
    \begin{align*}
        g(E, \mathbf{\Delta}; M) &= \max_{\mathcal{U} : |\mathcal{U}| = M} \sum_{u \in \mathcal{U} = \{u_1, \dots, u_M\}} 1 + \delta - |u|\\
        &= \max_{\mathcal{U} : |\mathcal{U}| = M} M(1 + \delta) - \sum_{u \in \mathcal{U} = \{u_1, \dots, u_M\}}|u|.
    \end{align*}
    We will further bound this by relaxing $|u_j| = x_j$ to real numbers. Let $y_j = \frac{x_j (x_j - 1)}{2}$, which implies $x_j = \frac{1 + \sqrt{1 + 8y_j}}{2}$. Then, since all $\binom{d}{2}$ edges must be covered by $\mathcal{U}$,
    \begin{align*}
        g(E, \mathbf{\Delta}; M) &\leq \max_{\substack{y_1, \dots, y_M \geq 1\\ \sum_{j = 1}^M y_j \geq \binom{d}{2}}} M(1 + \delta) - \sum_{j = 1}^M \frac{1 + \sqrt{1 + 8y_j}}{2}.
    \end{align*}
    We are maximizing a convex function over a convex region, so the maximum will be achieved at a vertex or at infinity. Since the function approaches $-\infty$ as any $y_j \to \infty$, the maximum must be achieved at a vertex. Each vertex results in the same function value, so select $y_1 = \dots = y_{M - 1} = 1$ and $y_M = \binom{d}{2} - M + 1$ to get
    \begin{equation}\label{eqn:convex_relax_upper_bound}
        g(E, \mathbf{\Delta}; M) \leq M(\delta - 1) + 2 - \frac{1 + \sqrt{1 + 8(\binom{d}{2} - M + 1)}}{2}.
    \end{equation}
    This bound is convex in $M$, so the maximum value of the bound over $d \leq M \leq \binom{d}{2}$ is achieved at $M = d$ or $M = \binom{d}{2}$, which means
    \begin{align*}
        \max_{d \leq M \leq \binom{d}{2}} g(E, \mathbf{\Delta}; M) &\leq \max\left\{d\delta - 2d + 3, -\binom{d}{2} + \binom{d}{2}\delta\right\}.
    \end{align*}
    Since the ``$d - 1$ clique plus star'' cover $\mathcal{U} = \{\{1, \dots, d - 1\}\} \cup \{\{1, d\}, \{2, d\}, \dots, \{d - 1, d\}\}$ achieves value $d\delta - 2d + 3$ and the individual edges cover $\mathcal{U} = E$ achieves value $-\binom{d}{2} + \binom{d}{2}\delta$, it follows that $g(E, \mathbf{\Delta}; d) = d\delta - 2d + 3$ and $g\left(E, \mathbf{\Delta}; \binom{d}{2}\right) = -\binom{d}{2} + \binom{d}{2}\delta$, so 
    \begin{equation}\label{eqn:gmax_geq_d}
        \max_{d \leq M \leq \binom{d}{2}} g(E, \mathbf{\Delta}; M) = \max\left\{d\delta - 2d + 3, -\binom{d}{2} + \binom{d}{2}\delta\right\}.
    \end{equation}
    We will next show that 
    \begin{equation}\label{eqn:hdef}
        h(d) \triangleq \max_{2 \leq M \leq d} g(E, \mathbf{\Delta}; M) = d\delta - 2d + 3,
    \end{equation}
    by induction on $d$. When $d = 3$, the only valid cover is $\mathcal{U} = \{\{1, 2\}, \{1, 3\}, \{2, 3\}\}$, achieving the value $3\delta - 3$. Now, assuming $h(d - 1) = (d - 1) \delta - 2(d - 1) + 3$, we wish to show that $h(d) = d\delta - 2d + 3$. Consider any candidate cover $\mathcal{U}$ of the clique spanned by $V_d = \{1, \dots, d\}$ such that $2 \leq |\mathcal{U}| \leq d - 1$. The case $|\mathcal{U}| = d$ is already handled by $g(E, \mathbf{\Delta}; d) = d\delta - 2d + 3$. We will first argue that we can assume there exists a vertex $v \in V_d$ satisfying
    \begin{itemize}
        \item $V_d \setminus \{v\} \notin \mathcal{U}$
        \item $|\{u \in \mathcal{U} : |u| = 2, v \in u\}| \leq 1$.
    \end{itemize}
    Suppose by way of contradiction that such a vertex $v$ does not exist, i.e., each $v \in V_d$ has either $V_d \setminus \{v\} \in \mathcal{U}$ or $|\{u \in \mathcal{U} : |u| = 2, v \in u\}| \geq 2$. Let $T$ denote the set of nodes $v \in V_d$ satisfying $|\{u \in \mathcal{U} : |u| = 2, v \in u\}| \geq 2$, and $W$ denote the set of nodes $v \in V_d$ satisfying $V_d \setminus \{v\} \in \mathcal{U}$. We have that $|W| + |T| \geq d$. Let $\ell = |\{u \in \mathcal{U}: |u| = 2\}$ be the number of size-2 sets in $\mathcal{U}$. The total degree spanned by all these edges is $2\ell$, and the total size-2 set degree associated with nodes in $T$ is at least $2|T|$, so $\ell \geq |T|$. Since
    \begin{equation*}
        d \leq |T| + |W| \leq \ell + |W| \leq |\mathcal{U}| \leq d - 1,
    \end{equation*}
    we arrive at a contradiction.
    



    Returning, for a vertex $v$ satisfying the above conditions, consider the cover $\mathcal{U}' \triangleq \{u \setminus \{v\} : u \in \mathcal{U}, |u \setminus \{v\}| \geq 2\}$ of the clique spanned by $V_d \setminus \{v\}$. Notice that $2 \leq |\mathcal{U}'| \leq d - 1$ since $V_d \setminus \{v\} \notin \mathcal{U}$. Then,
    \begin{align*}
        \left[\sum_{u \in \mathcal{U}} 1 + \delta - |u|\right] - \left[\sum_{u \in \mathcal{U}'} 1 + \delta - |u|\right]
        &= \left[\sum_{u \in \mathcal{U}, v \in u, |u| > 2} -1\right] + \left[\sum_{u \in \mathcal{U}, v \in u, |u| = 2} -1 + \delta\right]\\
        &= -|\{u \in \mathcal{U} : v \in u\}| + \delta \cdot |\{u \in \mathcal{U} : |u| = 2, v \in u\}|\\
        &\leq -2 + \delta,
    \end{align*}
    where we applied $|\{u \in \mathcal{U} : |u| = 2, v \in u\}| \leq 1$ and $|\{u \in \mathcal{U} : v \in u\}| \geq 2$ to achieve the last inequality. (If $|\{u \in \mathcal{U} : v \in u\}| = 1$, then that single element $u$ would necessarily equal $V_d$, which would be a contradiction.) Finally, since $2 \leq |\mathcal{U}'| \leq d - 1$, the induction hypothesis implies
    \begin{equation*}
        \left[\sum_{u \in \mathcal{U}'} 1 + \delta - |u|\right] \leq (d - 1)\delta - 2(d - 1) + 3 = d\delta -2d + 3 + 2 - \delta,
    \end{equation*}
    so
    \begin{align*}
        \left[\sum_{u \in \mathcal{U}} 1 + \delta - |u|\right] \leq d \delta - 2d + 3,
    \end{align*}
    proving Equation \ref{eqn:hdef}. Combining Equation \ref{eqn:gmax_geq_d} and Equation \ref{eqn:hdef} yields the desired result. 
\end{proof}

\subsection{Proof of Lemma~\ref{lemma:optimal_star_cover}}
\begin{proof}
    The case $d = 2$ is immediate. When $d \geq 3$, we can greedily reduce any cover $\mathcal{U}$ to individual edges to achieve a maximum value. Consider any $u \in \mathcal{U}$, with $\ell = |u| \geq 3$. Then, $u$ covers $\ell - 1$ edges. Without loss of generality, say that $u = \{1, \dots, \ell\}$. If we replace $u$ with the collection of edges $\{1, j\}$ for $2 \leq j \leq \ell$, then the objective value
    \begin{equation*}
        \sum_{u \in \mathcal{U}} 1 + \mathbf{\Delta}_{|u|} - |u|
    \end{equation*}
    increases by
    \begin{align*}
        (\ell - 1)(-1 + \delta) - (1 + \mathbf{\Delta}_\ell - \ell) &\geq (\ell - 1)(-1 + \delta) - (1 + \delta - \ell)\\
        &= (\ell - 2) \delta\\
        &\geq 0.
    \end{align*}
    Removing any newly redundant edges only further increases the objective since each edge contributes $-1 + \delta < 0$. Hence, we can maximize $g(E; \mathbf{\Delta})$ by converting to the all-edges cover, which implies $g(E; \mathbf{\Delta}) = -(d - 1) + (d - 1) \delta$.
\end{proof}

\end{document}